%% file: main.tex
\documentclass{article}

\usepackage[T1]{fontenc}
\usepackage{graphicx}

\usepackage{arxiv}
\usepackage{bm}
\usepackage{algorithm,algorithmic}
\usepackage{booktabs}
\usepackage{amsmath,amssymb,amsthm}
\usepackage{url}
\usepackage{doi}
\usepackage{hyperref}
\usepackage{multicol}

\usepackage{color}

\newcommand{\code}[1]{\texttt{#1}}
\newcommand{\nat}{\mathbb{N}}
\newcommand{\real}{\mathbb{R}}
\newcommand{\domain}{\mathop{\rm dom}}
\newcommand{\interior}{\mathop{\rm int}}
\newcommand{\combset}{\mathcal{C}}

\newcommand{\simplex}[1]{\mathcal{P}^{#1}}
\newcommand{\diamcomb}{D_{\infty}}
\newcommand{\relaxset}[1]{\mathop{\rm relax}(#1)}
\newcommand{\relaxalg}{\mathcal{B}}

\newcommand{\vecc}{\bm{c}}
\newcommand{\vecell}{\bm{\ell}}

\newcommand{\email}[1]{\texttt{#1}}
\newcommand{\orcid}[2]{
    \href{https://orcid.org/#1}
    {\includegraphics[scale=.05]{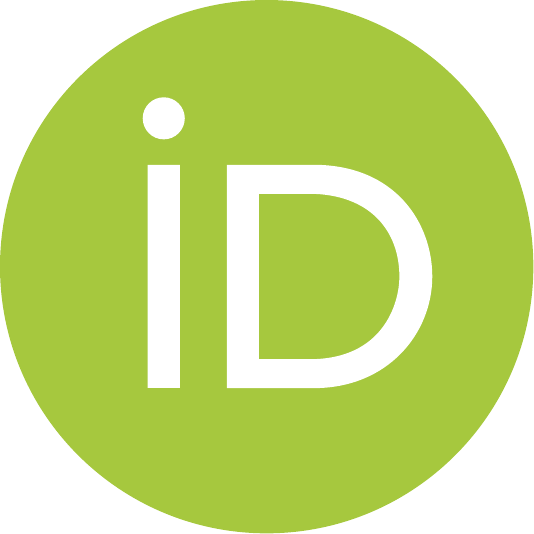} \hspace{1mm} #2}
}

\newtheorem{assumption}{Assumption}
\newtheorem{theorem}{Theorem}
\newtheorem{lemma}{Lemma}
\newtheorem{proposition}{Proposition}
\newtheorem{definition}{Definition}

\title{
    Online combinatorial linear optimization
    via a Frank-Wolfe based metarounding algorithm 
}

\author{
    \orcid{0000-0003-2277-6750}{Ryotaro Mitsuboshi} \\
    Kyushu University/RIKEN AIP \\
    \email{ryotaro.mitsuboshi@inf.kyushu-u.ac.jp} \\
    \And
    \orcid{0000-0002-1536-1269}{Kohei Hatano} \\
    Kyushu University/RIKEN AIP \\
    \email{hatano@inf.kyushu-u.ac.jp} \\
    \And
    \orcid{0000-0001-9542-2553}{Eiji Takimoto} \\
    Kyushu University \\
    \email{eiji@inf.kyushu-u.ac.jp}
}

\begin{document}
    \maketitle

    \begin{abstract}
        \input{abstract}

    \end{abstract}

    \input{introduction}
    \input{preliminary}

    \input{related_work}
    \input{contribution}

    \input{experiment}
    \input{conclusion}
    \input{acknowledgement}

    \bibliographystyle{splncs04}
    \bibliography{reference,hatano}
\end{document}

%% file: abstract.tex
Metarounding is an approach to convert an approximation algorithm for linear optimization 
over some combinatorial classes to 
an online linear optimization algorithm for the same class.  
We propose a new metarounding algorithm
under a natural assumption that a relax-based approximation algorithm exists for the combinatorial class.
Our algorithm is much more efficient in both theoretical and practical aspects.

%% file: introduction.tex
\section{Introduction}
Online decision-making problems using combinatorial objects arise in many applications, e.g., 
routing, advertisements, and resource allocation tasks. 
Examples of combinatorial objects include 
permutations~\cite{helmbold-warmuth:jmlr09},
$k$-sets~\cite{warmuth-kuzmin:jmlr08},
paths~\cite{gyorgy-etal:jmlr07},
matching~\cite{cesa-bianchi-lugosi:jcss12,helmbold-warmuth:jmlr09},
spanning trees~\cite{propp-wilson:ja98}
and so on.
%
%
%
%
%
%
%
%
%
%
%
%
Let $\combset\subset \nat^n_+$ be a combinatorial class, i.e., a set of combinatorial objects, 
where each combinatorial object is represented as a positive vector.
More formally, online combinatorial linear optimization over a combinatorial class $\combset$ 
is defined as a repeated game between 
a player and an environment. The protocol is the following: 
For each round $t = 1, 2, \dots$,
(i) the player chooses $\vecc_{t} \in \combset$.
Then, (ii) the environment reveals a loss vector $\vecell_{t} \in [0, 1]^{n}$ and 
(iii) the player incurs the loss $\vecc_t \cdot \vecell_t$.

The player's goal is to minimize the $\alpha$-regret defined as
\[
    R_{T}(\alpha) =
    \sum_{t=1}^{T} \bm{c}_{t} \cdot \bm{\ell}_{t}
    - \alpha \min_{\bm{c} \in \combset}
    \sum_{t=1}^{T} \bm{c} \cdot \bm{\ell}_{t},
\]
for as small $\alpha\geq 1$ as possible.

There are two approaches for online combinatorial linear optimization over $\combset$. 
The first approach is to design algorithms for individual classes~\cite{cesa-bianchi-lugosi:jcss12,helmbold-warmuth:jmlr09,warmuth-kuzmin:jmlr08,yasutake-etal:isaac11,yasutake-etal:cost11}.
The second approach is so called offline-to-online conversion, 
which uses an offline approximation algorithm over $\combset$
to construct an online algorithm over $\combset$ \cite{kakade+:sicomp09,garber:mor20,fujita+:alt13}.
An offline approximation algorithm takes $\bm{\ell}$ as input
and then outputs a vector $\bm{c} \in \combset$ satisfying
$
    \bm{c} \cdot \bm{\ell}
    \leq \alpha \min_{\bm{c}' \in \combset} \bm{c}' \cdot \bm{\ell}
$.
The offline-to-online conversion approach is generic in that 
once we design an efficient conversion algorithm, 
we can just use known approximation algorithms for different combinatorial classes.
For example, Garber proposed an algorithm
that calls the sub-routine $O(\ln T)$ times for each round $t$
and achieves $R_{T}(\alpha) = O(T^{2/3})$~\cite{garber:mor20}.
The paper also proposed 
an algorithm that calls the sub-routine $O(\sqrt{T} \ln T)$
per round and achieves $R_{T}(\alpha) = O(\sqrt{T})$.
The algorithm needs to know the number of total rounds $T$ and
the approximation ratio $\alpha$ a priori.
In addition, the number of sub-routine calls depends on $T$, so 
the conversion via the sub-routine may slow for a huge $T$.

\begin{figure}[t]
    \centering
    \includegraphics[keepaspectratio,scale=.45]{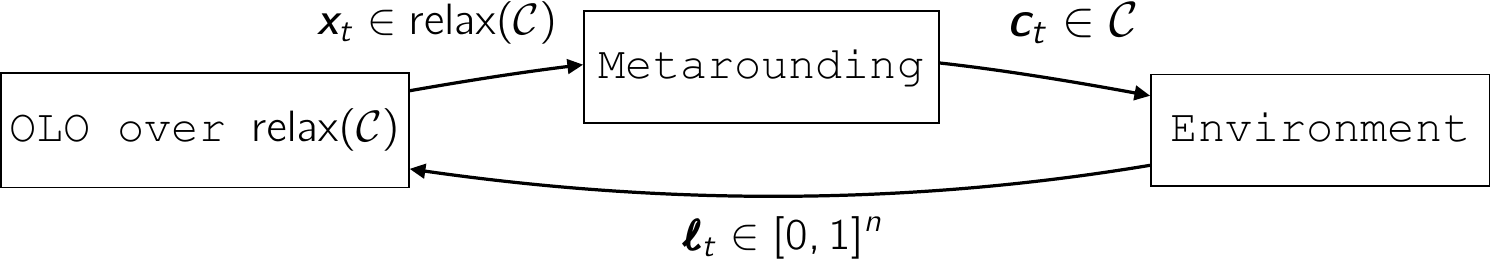}
    \caption{%
        A conceptual diagram for online combinatorial linear optimization %
        via metarounding. %
        In each round $t$, metarounding converts a vector $\bm{x}_{t}$ %
        into a combinatorial vector $\bm{c}_{t} \in \combset$. %
    }
    \label{fig:concept}
\end{figure}

Another offline-to-online conversion approach is to use metarounding\cite{carr-vempala:rsa02},
which uses a relaxation-based approximation algorithm. 
In fact, for many combinatorial classes $\combset$, 
there are relaxation-based approximation algorithms. 
Thus, the requirement is natural. 
Metarounding converts a vector $\bm{x} \in \relaxset{\combset}$ into
a combinatorial vector $\bm{c} \in \combset$,
where $\relaxset{\combset}$ is 
a convex set satisfying $\combset \subset \relaxset{\combset}$.
The relax-based approximation algorithms takes $\bm{\ell}$ as input
and then outputs a vector $\bm{c} \in \combset$ satisfying
$
    \bm{c} \cdot \bm{\ell}
    \leq \alpha \min_{\bm{p} \in \relaxset{\combset}} \bm{p} \cdot \bm{\ell}
$. The only difference between relaxation-based approximation 
and the approximation algorithm is 
the right-hand-side of the condition.
Using the relaxation-based approximation algorithm as sub-routine,
metarounding algorithms aim to output a vector $\bm{c} \in \combset$ randomly
to satisfy
$
    \mathbb{E}_{\bm{c}} [ \bm{c} ] \cdot \bm{\ell}
    \leq
    \alpha \bm{x} \cdot \bm{\ell}
$ for all $\bm{\ell}$.
Thus, combining the metarounding algorithm 
with an arbitrary online linear optimization algorithm 
over $\relaxset{\combset}$,
one can achieve 
$
    R_{T}(\alpha) = \alpha (1\text{-regret of the OLO algorithm})
$~\cite{fujita+:alt13}.
Thus, once an efficient metarounding algorithm is obtained,
we can fuse it to an OLO algorithm to obtain a better regret bound.
Figure~\ref{fig:concept} shows the concept of metarounding for
online combinatorial linear optimization over $\combset$.
Assuming the existence of such relax-based approximation algorithm,
Fujita et al.~\cite{fujita+:alt13} proposes a metarounding algorithm
that finds a distribution $\bm{\lambda} \in \simplex{\combset}$
satisfying $
    \mathbb{E}_{\bm{c}} [ \bm{c} ] \cdot \bm{\ell}
    \leq
    (\alpha + \epsilon) \bm{x} \cdot \bm{\ell}
$ for all $\bm{\ell}$
in $O(M^{2} \diamcomb^{2} n^{2} \ln(n) / \epsilon^{2})$
iterations, where $\diamcomb^{2}$ is
the maximal value of a combinatorial vector $\bm{c} \in \combset$
and $M$ is the input-dependent constant.

This paper proposes a new metarounding algorithm 
that finds a distribution $\bm{\lambda} \in \simplex{\combset}$ 
over $\combset$ that satisfies 
$
    \mathbb{E}_{\bm{c} \sim \bm{\lambda}} [ \bm{c} ] \cdot \bm{\ell}
    \leq (\alpha + \epsilon) \bm{x} \cdot \bm{\ell}
$ for all $\bm{\ell} \in [0, 1]^n$
in $O(\diamcomb^{2} \ln(n) / \epsilon^{2})$ rounds.
This guarantee is 
significantly better compared to the one by Fujita et al.~\cite{fujita+:alt13}.
Our algorithm is designed based on 
Frank-Wolfe algorithms~\cite{marguerite+:nrl56,jaggi:icml13}.
Table~\ref{tab:comparison} summarizes the previous works and our contribution.

Our technical contribution is two-fold. 
The first contribution is a re-formulation of the optimization problem for metarounding, 
based on a new observation (Proposition 1). 
This formulation enables us to derive algorithms with better iteration bounds. 
The second contribution is a new boosting-based metarounding algorithm. 
Our algorithm is similar to ERLPBoost\cite{warmuth+:alt08}, 
but ours employs a modified regularizer which is better suited for metarounding.
In fact, it can be viewed that our algorithm is a generalization of ERLPBoost for the approximation factor $\alpha>1$.
Furthermore, our analyses are based on those for Frank-Wolfe, 
which are completely different from those for ERLPBoost. 
As a result, our result significantly improves Fujita et al's 
on the number of approximation oracle calls per trial 
from $O(\diamcomb^{2} n^{2} \ln (n) / \epsilon^{2})$ 
to $O(\diamcomb^{2} \ln(n) / \epsilon^{2})$.
We also observe its significant improvements in the preliminary experiments.
\begin{table}[h]
    \centering
    \caption{%
        Comparison of the previous works. %
        Our assumption is the same as %
        the one in Fujita et al.~\cite{fujita+:alt13}. %
        Note that the work by Fujita et al. and ours aim to optimize %
        $(\alpha + \epsilon)$-regret $R_T(\alpha + \epsilon)$ %
        for arbitrarily small $\epsilon > 0$. %
    }
    \begin{tabular}{lclc}
        \toprule
        Algorithm
                  & \qquad
                    \begin{tabular}{c}
                      Approx. alg. \\
                      oracle
                    \end{tabular}
                  & \qquad
                    \begin{tabular}{c}
                      \# of oracle calls \\
                      per trial
                    \end{tabular}
                  & Regret \\
        \midrule
        Kalai et al.~\cite{kakade+:sicomp09}
                  & \qquad Any
                  & \qquad$O(T)$
                  & $O(T^{1/2})$ \\
        Garber~\cite{garber:mor20}
                  & \qquad Any
                  & \qquad$O(n^{2} \ln T)$
                  & $O(T^{2/3})$ \\
        Garber~\cite{garber:mor20}
                  & \qquad Any
                  & \qquad$O(n^{2} \sqrt{T} \ln T)$
                  & $O(T^{1/2})$ \\
        Fujita et al.~\cite{fujita+:alt13}
                  & \qquad relax-based 
                  & \qquad$
                        O(
                            \diamcomb^{2}
                            n^{2} \ln(n) / \epsilon^{2}
                        )
                    $ \qquad \qquad
                  & $O(T^{1/2})$ \\
        \textbf{This work}
                  & \qquad relax-based 
                  & \qquad$O(\diamcomb^{2} \ln(n) / \epsilon^{2})$
                  & $O(T^{1/2})$ \\
        \bottomrule
    \end{tabular}
    \label{tab:comparison}
\end{table}

%% file: preliminary.tex
\section{Preliminary}
Throughout this paper,
we use $\combset \subset \nat^{n}$ to represent the combinatorial set 
of our interest and let
$\diamcomb = \max_{\bm{c} \in \combset} \| \bm{c} \|_{\infty}$
as the maximal value of the entry of $\bm{c}$.
We use the notion $\bm{e}_{\bm{c}}$ to denote
the canonical basis vector over $\real^{\combset}$.
We denote $
    \simplex{n}
    = \{ \bm{p} \in [0, 1]^n \mid \| \bm{p} \|_{1} = 1 \}
$ as the probability simplex over $\real^{n}$.
For abbreviation, we often writes $\simplex{\combset}$ to represent
the probability simplex over a set $\combset$.
Given a set $\combset \subset \{0, 1\}^n$,
To introduce metarounding,
We also define 
the barrier function $I_{\combset} : \real^{n} \to \{0, \infty\}$
over a set $\combset \subset \real^{n}$ such that
$\bm{c} \in \combset$ iff $I_{\combset} (\bm{c}) = 0$.

\begin{definition}[Relax-based $\alpha$-approximation algorithm]
    Algorithm $\relaxalg$ is said to be
    a relax-based $\alpha$-approximation algorithm for $\combset$
    if $\relaxalg$ 
    takes $\bm{\ell} \in \real_{+}^{n}$ as input and then outputs
    a vector $\bm{c} \in \combset$ satisfying
    \begin{align*}
        \bm{c} \cdot \bm{\ell}
        \leq \alpha
        \min_{\bm{p} \in \relaxset{\combset}} \bm{p} \cdot \bm{\ell},
    \end{align*}
    where $\relaxset{\combset}$ is a convex set sastisfying
    $\combset \subset \relaxset{\combset}$.
\end{definition}

\begin{assumption}
    We assume the existence of a polynomial time 
    relax-based $\alpha$-approximation algorithm $\relaxalg$ 
    for $\combset$ with a relaxation set $\relaxset{\combset}$.
\end{assumption}

\begin{definition}[$\alpha$-metarounding]
    Algorithm $\mathcal{A}$ is said to be
    an $\alpha$-metarounding algorithm for $(\combset, \relaxset{\combset})$
    if $\mathcal{A}$, when given
    $\bm{x} \in \relaxset{\combset}$ as input,
    outputs a vector $\bm{c} \in \combset$ randomly
    such that for all $\bm{\ell} \in \real_{+}^{n}$,
    $
        \mathbb{E}_{\bm{c}} [ \bm{c} \cdot \bm{\ell} ]
        \leq
        \alpha \bm{x} \cdot \bm{\ell}
    $.
\end{definition}
The existence of a relax-based $\alpha$-approximation algorithm
implies the existence of an $\alpha$-metarounding algorithm~\cite{carr+:jrsa02}.

A function $f$ is said to be $\eta$-strongly convex
w.r.t. the norm $\| \cdot \|$ if
$
    f(\bm{y}) \geq f(\bm{x}) + (\bm{y} - \bm{x}) \cdot \nabla f(\bm{x})
    + \frac{\eta}{2} \| \bm{y} - \bm{x} \|^2
$
holds for all $\bm{x}, \bm{y} \in \domain f$.
Similarly, a function $f$ is said to be $\zeta$-smooth
w.r.t. the norm $\| \cdot \|$ if
$
    f(\bm{y}) \leq f(\bm{x}) + (\bm{y} - \bm{x}) \cdot \nabla f(\bm{x})
    + \frac{\zeta}{2} \| \bm{y} - \bm{x} \|^2
$
holds for all $\bm{x}, \bm{y} \in \domain f$.

Our analysis uses the Fenchel duality,
so we introduce the Fenchel conjugate function.
\begin{definition}[Fenchel conjugate]
    The Fenchel conjugate $f^\star$ of 
    a function $f : \mathbb{R}^k \to [-\infty, +\infty]$ is defined as
    \begin{align*}
        f^\star(\bm{\theta}) = \sup_{\bm{\mu} \in \mathbb{R}^{k}}
        \left( \bm{\mu} \cdot \bm{\theta} - f(\bm{\mu}) \right).
    \end{align*}
\end{definition}
The Fenchel conjugate of the max function $\bm{x} \mapsto \max_{i} x_{i}$ is
a barrier function of the probability simplex.
\begin{align*}
    f(\bm{x}) = \max_{i \in [k]} x_i
    \quad
    \iff
    \quad
    f^{\star}(\bm{y}) = I_{\simplex{k}}(\bm{y})
    = \begin{cases}
        0 & \bm{y} \in \simplex{k}, \\
        +\infty & \bm{y} \notin \simplex{k}.
    \end{cases}
\end{align*}
It is well known that if $f$ is a $1/\eta$-strongly convex function 
w.r.t. the norm $\|\cdot\|$ for some $\eta > 0$, 
$f^\star$ is an $\eta$-smooth function
w.r.t. the dual norm $\|\cdot\|_{\star}$. 
Further, if $f$ is a strongly convex function, 
the gradient vector of $f^\star$ is written as
$
    \nabla f^\star(\bm{\theta}) = 
    \arg \sup_{\bm{\mu} \in \mathbb{R}^{k}}
    \left( \bm{\mu} \cdot \bm{\theta} - f(\bm{\mu}) \right)
$.
One can find the proof of these properties 
here~\cite{borwein+:springer06,shalev-shwartz+:jml10}.
Now, we show the duality theorem~\cite{borwein+:springer06}.
\begin{theorem}[Fenchel duality theorem~\cite{borwein+:springer06}]
    \label{thm:strong_duality}
    Let $f : \mathbb{R}^{k} \to (-\infty, +\infty]$ and 
    $g : \mathbb{R}^{n} \to (-\infty, +\infty]$ be convex functions, 
    and a linear map $A : \mathbb{R}^{n} \to \mathbb{R}^{k}$. 
    Define the Fenchel problems
    \begin{align}
        \label{eq:strong_duality}
        \gamma = \inf_{\bm{d}} f(\bm{d}) + g(A^\top\bm{d}), 
        \qquad
        \rho   = \sup_{\bm{w}} -f^\star(-A \bm{w}) - g^\star(\bm{w}).
    \end{align}
    Then, $\gamma \geq \rho$ holds. 
    Further, $\gamma = \rho$ holds if\footnote{%
        For a convex set $\mathcal{S} \subset \mathbb{R}^{n}$, %
        $
            \interior(\mathcal{S}) = \{
                \bm{w} \in \mathcal{S} \mid
                \forall \bm{v} \in \mathbb{R}^{n},
                    \exists t > 0,
                    \forall \tau \in [0, t],
                    \bm{w} + \tau \bm{v} \in \mathcal{S}
            \}
        $ and 
        $
            \domain g - A^\top \domain f
            = \{ \bm{w} - A^\top \bm{d} \mid
                \bm{w} \in \domain g,
                \bm{d} \in \domain f
            \}
        $. %
    }
    $\bm{0} \in \interior \left(\domain g - A^\top \domain f\right)$. 
    Furthermore, points $\bar{\bm{d}} \in \simplex{m}_{\nu}$ and 
    $\bar{\bm{w}} \in \simplex{n}$ are optimal solutions 
    for problems in~(\ref{eq:strong_duality}) respectively 
    if and only if $-A \bar{\bm{w}} \in \partial f(\bar{\bm{d}})$ 
    and $\bar{\bm{w}} \in \partial g(A^\top \bar{\bm{d}})$.
\end{theorem}

%% file: related_work.tex
\section{Related Work}
Previous work for this setting~\cite{fujita+:alt13}
tries to solve the problem with a boosting-based approach
designed based on SoftBoost~\cite{warmuth+:nips07}.
The resulting algorithm named \code{MbB} 
(the shorthand for ``Metarounding-by-Boosting'')
is guaranteed to terminate 
in $O\left( M^{2} \diamcomb^{2} n^{2} \ln (n) / \epsilon^{2} \right)$ 
iterations,
where $M$ is the input-dependent constant.
They showed the efficiency of \code{MbB} by
numerical experiments against the ellipsoid method~\cite{carr+:jrsa02}.
There are two main advantages for \code{MbB}.
The first is that \code{MbB} does not need to know
the approximation ratio $\alpha$ a priori.
The other is that the iteration bound for \code{MbB} 
does not depend on the iteration $T$ of online algorithms.
Previous works take $O(\ln T)$~\cite{garber:mor20} 
to $O(T)$~\cite{carr+:jrsa02} iterations for
every round of online combinatorial optimization,
where $T$ is the number of total rounds for the online algorithm.

\subsection{Margin optimization boosting algorithms}
Boosting is a common machine learning technique that
combines multiple low-accuracy predictors to yield a highly accurate one.
Since the predictors to be combined are chosen from
a set of exponentially many predictors,
one cannot compute the combined predictor directory.
The main idea of boosting is to collect the most promising predictors
one by one until some condition is satisfied.

Gradient Boosting Machines~\cite{friedman:ann.stat.01},
such as XGBoost~\cite{tianqi+:kdd16} and LightGBM~\cite{ke+:nips17}, are
popular algorithms that minimize some loss functions.
For binary classification,
it is well known that a predictor
that maximizes the quantity so-called ``margin''
instead of minimizing loss functions
is known to have smaller losses for unseen data~\cite{schapire+:as98}.
Thus, some algorithms,
such as LPBoost~\cite{demiriz+:ml02},
SoftBoost~\cite{warmuth+:nips07}, and
ERLPBoost~\cite{warmuth+:alt08},
are proposed to optimize the margin. 

\subsection{Frank-Wolfe algorithms}
Frank-Wolfe algorithms (FWs for shorthand),
a.k.a. conditional gradient algorithms,
are the class of first-order optimization algorithms
for convex objectives over convex sets.
FW was originally invented by Frank \& Wolfe~\cite{marguerite+:nrl56} and then 
studied extensively~\cite{jaggi:icml13,pedregosa+:aistats20,tsuji+:icml22}.
FW algorithms are known to converge in $O(\eta / \epsilon)$ iterations
if the objective function is $\eta$-smooth w.r.t. some norm.

%% file: contribution.tex
\input{pseudo-code}

\section{Main idea}
First, we revisit the approach by Fujita et al.~\cite{fujita+:alt13}.
In \code{MbB}, they formulate an $\alpha$-metarounding
as an algorithm that solves the following optimization problem.
\begin{align}
    \label{eq:main-primal}
    \min_{\beta, \bm{\lambda}} \; \beta
    \qquad
    \text{s.t.} \quad
    & \sum_{\bm{c} \in \combset} \lambda_{\bm{c}} c_{i} \leq \beta x_{i},
        \quad \forall i \in [n], \\
    \nonumber
    & \bm{\lambda} \in \simplex{\combset}.
\end{align}
The Lagrange dual problem for Problem (\ref{eq:main-primal}) is
\begin{align}
    \label{eq:main-dual}
    \max_{\gamma, \bm{\ell}} \; \gamma
    \qquad
    \text{s.t.} \quad
    &
    \bm{\ell} \cdot \bm{c} \geq \gamma,
    \quad \forall \bm{c} \in \combset, \\
    \nonumber
    & \bm{\ell} \cdot \bm{x} = 1, \quad
    \bm{\ell} \geq \bm{0}.
\end{align}
Problem (\ref{eq:main-dual}) has $O(|\combset|)$ constraints,
so \code{MbB} uses the column-generation approach 
like soft margin boosting algorithms;
That is, starting from $\tilde{\combset} = \emptyset$,
they solve a problem under the constraints in (\ref{eq:main-dual})
over $\tilde{\combset}$ and then
append a new combinatorial vector $\bm{c} \in \combset$
to $\tilde{\combset}$ one by one until some conditions are met.
One can see that Problem (\ref{eq:main-dual}) can be seen
as so-called ``Edge minimization'' in boosting 
(See, e.g.,~\cite{demiriz+:ml02,warmuth+:nips07}).
In fact, setting $\bm{x} = \bm{1}$ in Problem (\ref{eq:main-dual}) corresponds
to margin optimization.
With this insight, \code{MbB} is proposed based on a boosting algorithm
named SoftBoost~\cite{warmuth+:nips07}, but the resulting guarantee is
not as good as the one in SoftBoost.

We want to use the state-of-the-art soft margin optimization
named ERLPBoost~\cite{warmuth+:alt08} to solve Problem~(\ref{eq:main-dual}),
which uses a normalized relative entropy regularizer,
but there are some issues.
First, 
constraint $\bm{\ell} \cdot \bm{x} = 1$ makes the direct application hard.
ERLPBoost succeeded because its feasible region is probability simplex,
while Problem (\ref{eq:main-primal}) is not.
One naive idea is to use unnormalized relative entropy 
instead of the normalized one.
This approach results in a similar result to Fujita et al. 
since the variable $\bm{\ell}$ is unbounded.
With these issues, we use a different approach,
a Frank-Wolfe-based approach similar to 
the one in Mitsuboshi et al.~\cite{mitsuboshi+:arxiv22}.
First, we convert Problem (\ref{eq:main-dual}) 
to an equivalent bounded problem and then use an ERLPBoost-like problem.
The resulting algorithm is a generalization of the soft margin boosting
in the following sense:
Our algorithm becomes ERLPBoost if we set $\bm{x} = \bm{1}$.


\subsection{Bounding the dual variable.}
Before getting into our main result, we prove the following proposition.
\begin{proposition}
    \label{prop:variable-bound}
    Let $M := \max \{ \frac{1}{x_{i}} \mid x_{i} \neq 0, i \in [n] \}$.
    There is an optimal solution $(\gamma^{\star}, \bm{\ell}^\star)$
    to Problem (\ref{eq:main-dual}) satisfying $\ell^{\star}_{i} \leq M$
    for all $i \in [n]$.
\end{proposition}
\begin{proof}
    Let $(\beta^{\star}, \bm{\lambda}^{\star})$,
    $(\gamma^{\star}, \bm{\ell}^{\star})$ be optimal solutions to
    Problem~(\ref{eq:main-primal}),~(\ref{eq:main-dual}), respectively.
    Without loss of generality, we can restrict the combinatorial set to
    $\combset^{\star} := \{ \bm{c} \in \combset \mid \lambda_{\bm{c}} > 0 \}$.
    We consider the following two cases.
    \begin{enumerate}
        \item $\forall i \in [n], x_{i} \neq 0$. This case,
            \[
                \bm{\ell}^{\star} \cdot \bm{x}
                > \sum_{j \neq i} \ell^{\star}_{j} x_{j} + M x_{i}
                > \sum_{j \neq i} \ell^{\star}_{j} x_{j}
                    + M \min_{k \in [n]} x_{k}
                \geq 1,
            \]
            which contradicts to
            the constraint $\bm{\ell}^{\star} \cdot \bm{x} = 1$,
            so that $\ell^{\star}_{i} \leq M$ for all $i \in [n]$.
        \item $\exists i \in [n], x_{i} = 0$. 
            Since $(\beta^{\star}, \bm{\lambda}^{\star})$ is a feasible solution,
            $
                \sum_{\bm{c} \in \combset^{\star}}
                \lambda^{\star}_{\bm{c}} c_{i}
                \leq \beta^{\star} x_{i} = 0
            $, which implies $\lambda^{\star}_{\bm{c}} = 0$ or $\bm{c}_{i} = 0$
            for all $\bm{c} \in \combset^{\star}$.
            By definition of $\combset^{\star}$, $\lambda^{\star}_{\bm{c}} > 0$
            for all $\bm{c} \in \combset^{\star}$, so
            we have $c_{i} = 0$ for all $\bm{c} \in \combset^{\star}$.
            Thus, $\ell^{\star}_{i} c_{i} = 0$ does not affect to
            the objective value $\gamma^{\star}$.
            So we can set $\ell^{\star}_{i} = 0$ without losing
            the optimality.
    \end{enumerate}
\end{proof}
As discussed in the above proof, 
the variables $\{\ell_{i} \mid x_{i} = 0, i \in [n] \}$ do not contribute to
the optimal value, so we assume that $x_{i} > 0$ for all $i \in [n]$
without loss of generality.
Furthermore, restricting $\bm{\ell} \in [0, M]^{n}$ does not affect
the optimal solution to Problem~(\ref{eq:main-dual}),
so we consider the following problem
\begin{align}
    \label{eq:main-dual-02}
    \max_{\gamma, \bm{\ell}} \; \gamma
    \qquad
    \text{s.t.} \quad
    &
    \bm{\ell} \cdot \bm{x} = 1, \quad
    \bm{\ell} \cdot \bm{c} \geq \gamma,
    \quad \forall \bm{c} \in \combset, 
    \quad \bm{0} \leq \bm{\ell} \leq \bm{M},
\end{align}
where $M$ is the value defined in Proposition~\ref{prop:variable-bound}.
To clarity the analysis, we rewrite the above problem by introducing
$\bar{\bm{\ell}} = (\ell_{1} x_{1}, \ell_{2} x_{2}, \dots, \ell_{n} x_{n})$ and
$
    \bar{\combset} = \{
        \bar{\bm{c}} \mid
        \bar{\bm{c}} = (c_{1}/x_{1}, c_{2}/x_{2}, \dots, c_{n}/x_{n}),
        \bm{c} = (c_{1}, c_{2}, \dots, c_{n}) \in \combset
    \}
$.
\begin{align}
    \label{eq:main-dual-03}
    \max_{\gamma, \bar{\bm{\ell}}} \; \gamma
    \qquad
    \text{s.t.} \quad
    &
    \bar{\bm{\ell}} \cdot \bar{\bm{c}} \geq \gamma,
    \quad \forall \bar{\bm{c}} \in \bar{\combset}, 
    \quad \bar{\bm{\ell}} \in \simplex{n}.
\end{align}
Obviously, Problem (\ref{eq:main-dual-03}) is 
equivalent to (\ref{eq:main-dual-02}).
Thus, the following sections focus on solving Problem (\ref{eq:main-dual-03}).




\subsection{The regularization}
Regularizing a smooth function to a convex objective function
does not make the objective smooth
while regularizing a strongly convex function
makes it strongly convex.
So, we use a strongly convex regularizer 
$
    \Delta(\bar{\bm{\ell}})
    = \sum_{i=1}^{n} \bar{\ell_{i}} \ln \frac{\bar{\ell_{i}}}{1/n}
$ to the objective function of~(\ref{eq:main-dual-03})
then consider the primal problem.
We solve Problem (\ref{eq:main-dual-03}) by ERLPBoost-based algorithm,
shown in Algorithm~\ref{alg:metarounding-erlpboost}.
Our algorithm aims to find an $\epsilon/2$-approximate solution
to the following problem.
\begin{align}
    \label{eq:main-dual-entropy}
    \max_{\gamma, \bm{\ell}} \; \gamma
    - \frac{1}{\eta}
    \Delta(\bar{\bm{\ell}})
    \qquad
    \text{s.t.} \quad
    \bar{\bm{\ell}} \cdot \bar{\bm{c}} \geq \gamma, \quad
    \forall \bar{\bm{c}} \in \bar{\combset},
    \quad \bar{\bm{\ell}} \in \simplex{n}.
\end{align}
Note that, by definition, $\Delta( \bar{\bm{\ell}} ) \leq \ln n$ holds
for all $\bar{\bm{\ell}} \in \simplex{n}$.

Let $I_{\simplex{n}}$ be the barrier function for $\simplex{n}$,
$J(\bm{\theta}) = \max_{\bm{c} \in \bar{\combset}} \theta_{\bm{c}}$
be the max function over $\real^{\bar{\combset}}$, and
let $H = I_{\simplex{n}} + \frac{1}{\eta} \Delta$.
With these notations, one can rewrite Problem (\ref{eq:main-dual-entropy}) as
\begin{align}
    \label{eq:main-dual-entropy-02}
    \max_{\bar{\bm{\ell}}} \;
        - J(- \bar{\bm{\ell}}^\top \bar{C})
        - I_{\simplex{n}}(\bar{\bm{\ell}})
        - \frac{1}{\eta} \Delta(\bar{\bm{\ell}})
    =: 
    \max_{\bar{\bm{\ell}}} \;
    - J(- \bar{\bm{\ell}}^\top \bar{C})
    - H(\bar{\bm{\ell}}),
\end{align}
where 
$
\bar{C} =
    \begin{bmatrix}
        \bar{\bm{c}}_1 & \bar{\bm{c}}_{2} & \dots
    \end{bmatrix}
    \in \real^{n \times \bar{\combset}}
$ is 
the matrix whose column vectors are elements of $\bar{\combset}$.
By Theorem~\ref{thm:strong_duality},
Problem (\ref{eq:main-dual-entropy-02}) is equivalent to the following
dual problem\footnote{%
    One can verify the equality %
    by setting $f = J^{\star}$, $g = H^{\star}$, and $A = \bar{C}^\top$. %
}.
\begin{align}
    \label{eq:main-primal-entropy}
    \min_{\bm{\lambda}} \;
    J^{\star}(\bm{\lambda}) + H^{\star}(\bar{C} \bm{\lambda})
    =
    \min_{\bm{\lambda} \in \simplex{\bar{\combset}}} \;
    H^{\star}(\bar{C} \bm{\lambda})
\end{align}
Since $H$ is $1/\eta$-strongly convex w.r.t. $L_1$-norm,
its conjugate $H^{\star}$ is $\eta$-smooth w.r.t. $L_\infty$-norm.
The objective function $H^{\star}$ is
a smoothified version of the one in Problem (\ref{eq:main-primal}).

The smoothness of the objective function is a key property for our analysis,
so we aim to solve Problem (\ref{eq:main-primal-entropy}) instead of
Problem (\ref{eq:main-primal}).
Algorithm~\ref{alg:metarounding-erlpboost} summarizes our method,
which maintains the quantity $\epsilon_{k}$
defined in Line~\ref{alg-line:epsilon}.
\begin{lemma}
    \label{lem:justification}
    Let $\bm{x} \in \relaxset{\combset}$ be 
    an input to Algorithm~\ref{alg:metarounding-erlpboost} and
    let $
        \epsilon_{k}
        = H^{\star}(\bar{C}^{(k)} \bm{\lambda})
        - \max_{j \in [k]} \bar{\bm{c}}_j \cdot \bar{\bm{\ell}}_{j-1}
    $ be the optimality gap 
    defined in Line~\ref{alg-line:epsilon}
    of Algorithm~\ref{alg:metarounding-erlpboost}.
    If $\eta \geq 2 \ln(n) / \epsilon$, 
    $\epsilon_{k} \leq \epsilon / 2$ implies 
    $
        \mathbb{E}_{\bm{c} \sim \bm{\lambda}_{k}} [ \bm{c} ] \cdot \bm{\ell}
        \leq (\alpha + \epsilon) \bm{x} \cdot \bm{\ell}
        $ for all $\bm{\ell} \in [0, M]^{n}$.
\end{lemma}
\begin{proof}
    Recall that
    $
        \bar{\bm{\ell}}_{k}
        = \nabla H^\star (\bar{C}^{(k)} \bm{\lambda}_{k})
        = \arg \max_{\bar{\bm{\ell}} \in \simplex{n}}
        \bar{\bm{\ell}}^{\top} \bar{C}^{(k)} \bm{\lambda}_{k}
        - \frac{1}{\eta} \Delta (\bar{\bm{\ell}}).
    $
    By definition of $H^{\star}$, 
    the following holds for all $\bm{\ell} \in [0, M]^{n}$.
    \begin{align*}
        H^{\star}(\bar{C}^{(k)} \bm{\lambda}_{k})
        = \; & \bar{\bm{\ell}}_{k}^{\top} \bar{C}^{(k)} \bm{\lambda}_{k}
        - \frac{1}{\eta} \Delta(\bar{\bm{\ell}}_{k}) \\
        \geq \; & \bar{\bm{\ell}}^\top \bar{C}^{(k)} \bm{\lambda}_{k}
            - \frac{1}{\eta} \Delta (\bar{\bm{\ell}}) 
            \geq \bm{\ell}^\top C^{(k)} \bm{\lambda}_{k}
            - \frac{1}{\eta} \ln n 
        \geq 
            \mathbb{E}_{\bm{c} \sim \bm{\lambda}_{k}}
            \left[ \bm{\ell} \cdot \bm{c} \right]
            - \frac{\epsilon}{2}.
    \end{align*}
    On the other hand,
    by the assumption of the $\relaxalg$,
    we have
    \begin{align*}
        \max_{j \in [k]} \bar{\bm{\ell}}_{j} \cdot \bar{\bm{c}}_{j+1}
        = \max_{j \in [k]} \bm{\ell}_{j} \cdot \bm{c}_{j+1}
        \leq & \max_{j \in [k]}
            \alpha \min_{\bm{p} \in \relaxset{\combset}}
            \bm{\ell}_{j} \cdot \bm{p} \\
        \leq & \alpha \max_{j \in [k]} \bm{\ell}_{j} \cdot \bm{x}
        = \alpha \max_{j \in [k]} \|\bar{\bm{\ell}}_{j}\|_{1} = \alpha,
    \end{align*}
    where the second inequality comes from 
    the fact that $\bm{x} \in \relaxset{\combset}$.
    Combining the above relations, we have
    \[
        \epsilon_{k}
        = H^\star( \bar{C}^{(k)} \bm{\lambda}_{k} )
        - \max_{j \in [k]} \bar{\bm{c}}_{j+1} \cdot \bar{\bm{\ell}}_{j}
        \geq \mathbb{E}_{\bm{c} \sim \bm{\lambda}_{k}}
        \left[ \bm{\ell} \cdot \bm{c} \right]
        - \frac{\epsilon}{2}
        - \alpha, \qquad \forall \bm{\ell} \in [0, M]^{n}.
    \]
    Thus, $\epsilon_{k} \leq \epsilon/2$ implies 
    $
        \mathbb{E}_{\bm{c} \sim \bm{\lambda}_{k}}
        \left[ \bm{\ell} \cdot \bm{c} \right]
        \leq \alpha + \epsilon
        = (\alpha + \epsilon) \| \bar{\bm{\ell}} \|_{1}
        = (\alpha + \epsilon) \bm{x} \cdot \bm{\ell}
    $
    for all $\bm{\ell} \in [0, M]^{n}$,
    which is the relation we desired.
\end{proof}
Lemma~\ref{lem:justification} assures that
the output $\bm{\lambda}^{\star}$ of
Algorithm~\ref{alg:metarounding-erlpboost}
is an $\epsilon$-approximate solution for Problem (\ref{eq:main-primal}).
Thus, we can focus on proving how quickly the optimality gap converges.
The following lemma shows the convergence rate.
\begin{lemma}
    \label{lem:recurrence}
    For all $\{ \mu_{k} \}_{k \geq 1} \subset [0, 1]$,
    the sequence $\{ \epsilon_{k} \}_{k \geq 1}$ 
    generated by Algorithm~\ref{alg:metarounding-erlpboost} satisfies
    the following relation.
    \begin{align}
        \label{eq:recurrence}
        \forall k \geq 1, \quad
        \epsilon_{k+1} \leq \left(1 - \mu_{k} \right) \epsilon_{k}
        + 2 \eta \mu_{k}^{2} M^{2} \diamcomb^{2}.
    \end{align}
\end{lemma}
\begin{proof}
    First of all, $\bm{\lambda}_{k+1}$ is chosen to optimize
    \begin{align*}
        \min_{\bm{\lambda} \in \simplex{k+1}}
        H^{\star} ( \bar{C}^{(k+1)} \bm{\lambda} ),
        \qquad
        \bar{C}^{(k+1)} = \begin{bmatrix}
            \bar{\bm{c}}_{1} & \bar{\bm{c}}_{2} & \dots & \bar{\bm{c}}_{k+1}
        \end{bmatrix}
        \in [0, M\diamcomb]^{n \times (k+1)},
    \end{align*}
    so that
    $
        H^{\star}(\bar{C}^{(k+1)} \bm{\lambda}_{k+1})
        \leq H^{\star}(\bar{C}^{(k+1)} \bm{\lambda})
    $ holds
    for all interior points $
        \bm{\lambda}
        = \bm{\lambda}_{k} + \mu_{k} \left(
            \bm{e}_{\bar{\bm{c}}_{k+1}} - \bm{\lambda}_{k}
        \right)
    $ with $\mu_{k} \in [0, 1]$.
    Thus, using the $\eta$-smoothness of $H^{\star}$,
    \begin{align*}
        \epsilon_{k} - \epsilon_{k+1}
        \geq \; &
        H^{\star} (\bar{C}^{(k)} \bm{\lambda}_{k})
        - H^{\star} (\bar{C}^{(k+1)} \bm{\lambda}) \\
        \geq \; &
            - \mu_{k}
            \nabla H^{\star} (\bar{C}^{(k+1)} \bm{\lambda}_{k})^{\top}
            \bar{C}^{(k+1)} \left( \bm{e}_{\bar{\bm{c}}_{k+1}} - \bm{\lambda}_{k} \right) \\
                &
            - \frac{\eta}{2} \mu_{k}^{2}
            \|
                \bar{C}^{(k+1)} (\bm{e}_{\bar{\bm{c}}_{k+1}} - \bm{\lambda}_{k})
            \|_{\infty}^{2} \\
        \geq \; &
            - \mu_{k}
            \bar{\bm{\ell}}_{k}
            \bar{C}^{(k+1)} \left( \bm{e}_{\bar{\bm{c}}_{k+1}} - \bm{\lambda}_{k} \right)
            - 2 \eta \mu_{k}^{2} M^{2} \diamcomb^{2} \\
        = \; &
            - \mu_{k} \bar{\bm{\ell}}_{k} \cdot \bar{\bm{c}}_{k+1}
            + \mu_{k} \bar{\bm{\ell}}_{k}^{\top} \bar{C}^{(k)} \bm{\lambda}_{k}
            - 2 \eta \mu_{k}^{2} M^{2} \diamcomb^{2} \\
        \geq \; &
            - \mu_{k} \bar{\bm{\ell}}_{k} \cdot \bar{\bm{c}}_{k+1}
            + \mu_{k} \left[
                \bar{\bm{\ell}}_{k}^{\top} \bar{C}^{(k)} \bm{\lambda}_{k}
                - \frac{1}{\eta} \Delta(\bar{\bm{\ell}}_{k})
            \right] 
            - 2 \eta \mu_{k}^{2} M^{2} \diamcomb^{2} \\
        = \; &
            - \mu_{k}
            \bar{\bm{\ell}}_{k} \cdot \bar{\bm{c}}_{k+1}
            + \mu_{k} H^{\star}( \bar{C}^{(k)} \bm{\lambda}_{k} )
            - 2 \eta \mu_{k}^{2} M^{2} \diamcomb^{2} \\
        \geq \; & \mu_{k} \epsilon_{k} - 2 \eta \mu_{k}^{2} M^{2} \diamcomb^{2}.
    \end{align*}
    Rearraging terms, we get the recurrence relation.
\end{proof}
\begin{theorem}
    \label{thm:convergence-rate}
    Algorithm~\ref{alg:metarounding-erlpboost} terminates
    in $
    K = O\left( M^{2}\diamcomb^{2}\ln (n) / \epsilon^{2} \right)
    $ iterations.
\end{theorem}
\begin{proof}
    Since Lemma~\ref{lem:recurrence} does not specify the sequence
    $\{\mu_{k}\}_{k \geq 1} \subset [0, 1]$,
    we choose $\mu_{k} = \frac{2}{k+2}$.
    Now we will show the following by induction on $k$.
    \begin{align}
        \label{eq:induction-hypothesis}
        \epsilon_{k} \leq \frac{8 \eta M^{2} D_{\infty}^{2}}{k + 2}
    \end{align}
    For the base case $k = 1$, $
    \epsilon_{1} \leq \left(1 - \frac{2}{0 + 2}\right) \epsilon_{0}
        + 2 \eta \left( \frac{2}{0 + 2} \right)^{2} M^{2} \diamcomb^{2}
        \leq \frac{8 \eta M^{2} \diamcomb^{2}}{1 + 2}
    $ holds. Thus Eq. (\ref{eq:induction-hypothesis}) holds for $k = 1$.
    Now, assume that Eq. (\ref{eq:induction-hypothesis}) holds for $k \geq 1$.
    By the recurrence formula (\ref{eq:recurrence}) and
    the induction hypothesis, we have
    \begin{align*}
        \epsilon_{k+1}
        \leq & \left(1 - \frac{2}{k+2}\right)
        \frac{8 \eta M^{2} \diamcomb^{2}}{k + 2}
        + 2 \eta \left( \frac{2}{k+2} \right)^{2} M^{2} \diamcomb^{2} \\
        = & \frac{k}{k+2} \frac{8 \eta M^{2} \diamcomb^{2}}{k + 2}
        + 2 \eta \left( \frac{2}{k+2} \right)^{2} M^{2} \diamcomb^{2} \\
        = & \frac{8 \eta M^{2} \diamcomb^{2}}{k+2} \left[
            \frac{k}{k+2} + \frac{1}{k+2}
        \right]
        \leq \frac{8 \eta M^{2} \diamcomb^{2}}{k+2} \frac{k+2}{k+3}
        = \frac{8 \eta M^{2} \diamcomb^{2}}{k+3}
    \end{align*}
    Thus, Ineq. (\ref{eq:induction-hypothesis}) holds for all $k$.

    Setting the right-hand-side of (\ref{eq:induction-hypothesis}) 
    to be at most $\epsilon / 2$,
    we $\epsilon_{K} \leq \epsilon / 2$ is achieved 
    after $
        K \geq 16 \eta M^{2} \diamcomb^{2}/\epsilon
        \geq 32M^{2} \diamcomb^{2} \ln (n) / \epsilon^{2}
    $ iterations.
\end{proof}
By Theorem~\ref{thm:convergence-rate},
$\epsilon_{K} \leq \frac{\epsilon}{2}$ holds
after $K = O\left(M^{2} \diamcomb^{2}\ln (n) / \epsilon^{2}\right)$ iterations.
Further, Lemma~\ref{lem:justification} guarantees
that the output $\bm{\lambda}_{K} \in \simplex{\combset}$ satisfies
$
    \mathbb{E}_{\bm{c} \sim \bm{\lambda}} [ \bm{c} ] \cdot \bm{\ell}
    \leq (\alpha + \epsilon) \bm{x} \cdot \bm{\ell}
    $ for all $\bm{\ell} \in [0, M]^{n}$.
Therefore, Algorithm~\ref{alg:metarounding-erlpboost} is
a matarounding algorithm.
We remark that if $\bm{x} = \bm{1}$, 
the convergence rate is the same as in ERLPBoost.
In this sense, our algorithm can be seen as a generalized version of ERLPBoost.
Furthermore, our analysis uses the Frank-Wolfe-like analysis, 
significantly simplifying the one ERLPBoost.

%% file: pseudo-code.tex
\begin{algorithm}[t]
    \caption{Metarounding by ERLPBoost}
    \label{alg:metarounding-erlpboost}
    \begin{algorithmic}[1]
        \REQUIRE{%
            $\bm{x} \in \relaxset{\combset}$, %
            accuracy $\epsilon > 0$, %
            and a relax-based $\alpha$-approx. 
            algorithm $\relaxalg$.%
        }
        \STATE{%
            Initialize $
                \bar{\bm{\ell}}_0
                = \frac{1}{n} \bm{x}, \bm{\ell}_{0} = \frac{1}{n} \bm{1}
            $ and set $H^{\star}(\bar{C}^{(0)} \bm{\lambda}_{0}) = \infty$, %
            $\eta = 2 \ln(n) / \epsilon$.
        }
        \FOR{$k = 0, 1, 2, \dots, K$}
            \STATE{%
                Call the approximation algorithm to obtain %
                $\bm{c}_{k+1} \in \relaxalg(\bm{\ell}_{k})$. \\
                Let $\bar{\bm{c}}_{k+1} = (c_{1}/M, c_{2}/M, \dots, c_{n}/M)$. %
            }
            \IF{%
                $
                    \epsilon_{k}
                    :=
                    H^{\star}(\bar{C}^{(k)} \bm{\lambda}_{k})
                    - \max_{j \leq k}
                    \bar{\bm{c}}_{j+1} \cdot \bar{\bm{\ell}}_{j}
                    \leq \epsilon / 2
                $
                \label{alg-line:epsilon}
            }
                \STATE{Set $K = k$. \textbf{break.}}
            \ENDIF
            \STATE{%
                Set
                $
                    \bm{\lambda}_{k+1}
                    \gets
                    \arg \min_{\bm{\lambda} \in \simplex{k+1}}
                    H^{\star} (\bar{C}^{(k+1)} \bm{\lambda})
                $, where %
                $
                    \bar{C}^{(k+1)} =
                    \begin{bmatrix}
                        \bar{\bm{c}}_{1} &
                        \bar{\bm{c}}_{2} &
                        \dots &
                        \bar{\bm{c}}_{k+1}
                    \end{bmatrix}
                $.
            }
            \STATE{%
                Set %
                $
                    \bar{\bm{\ell}}_{k+1}
                    = \nabla H^{\star} (\bar{C}^{(k+1)} \bm{\lambda}_{k+1})
                    = \arg \max_{\bar{\bm{\ell}} \in \simplex{n}}
                    \bar{\bm{\ell}}^\top \bar{C}^{(k+1)} \bm{\lambda}_{k+1}
                    - \frac{1}{\eta} \Delta (\bar{\bm{\ell}})
                    $ and $
                    \bm{\ell}_{k+1}
                    = (\bar{\ell}_{k+1,1}/x_{1},
                    \bar{\ell}_{k+1,2}/x_{2},
                    \dots,
                    \bar{\ell}_{k+1,n}/x_{n})
                $%
            }
        \ENDFOR
        \STATE{%
            Solve Problem (\ref{eq:main-primal}) over
            $\bm{c}_{1}, \bm{c}_{2}, \dots, \bm{c}_{K+1}$ and %
            let the optimal solution be $\bm{\lambda}^{\star}$. %
        }
        \ENSURE{$\bm{\lambda}^{\star} \in \simplex{\combset}$}
    \end{algorithmic}
\end{algorithm}

%% file: experiment.tex
\section{Experiment}
\begin{figure}[t]
    \centering
    \begin{minipage}[t]{.49\hsize}
        \centering
        \includegraphics[keepaspectratio,scale=.30]{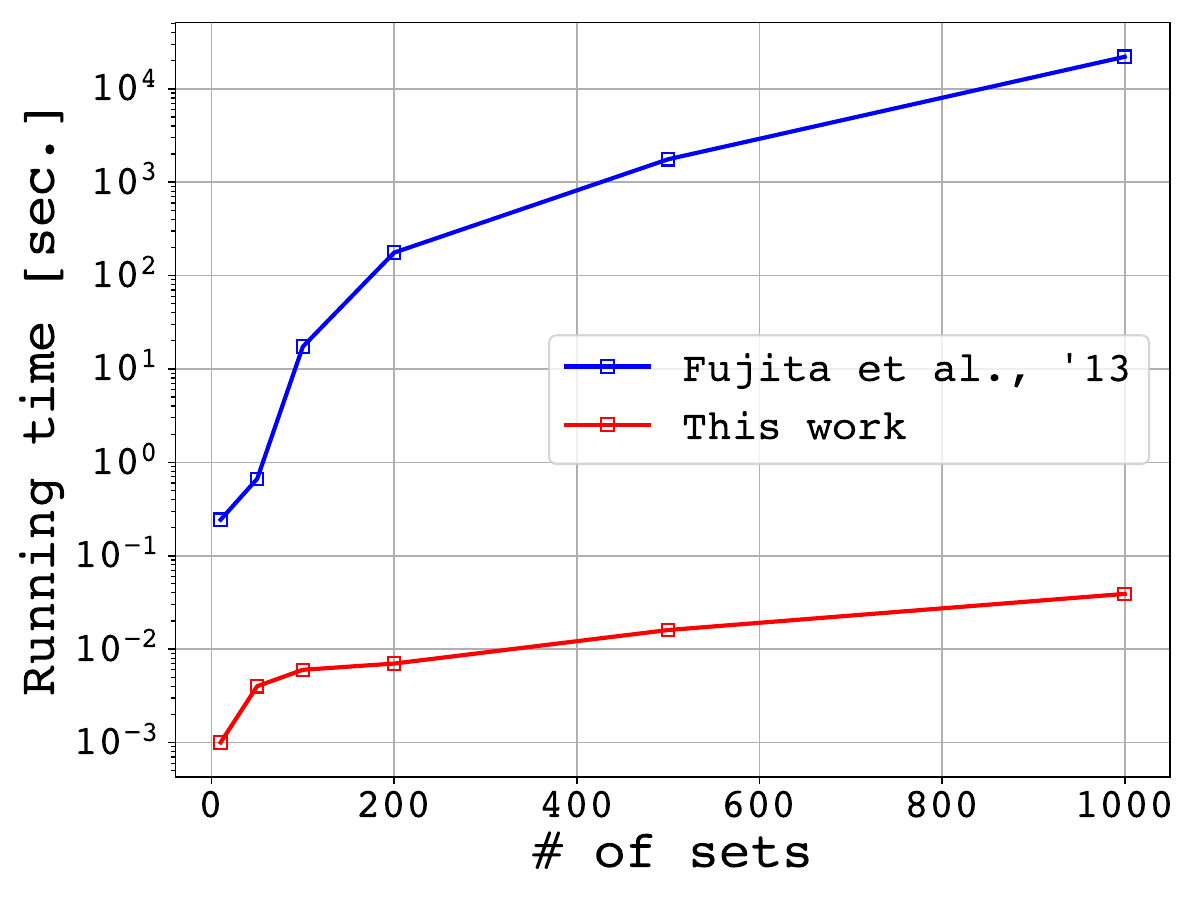}
    \end{minipage}
    \begin{minipage}[t]{.49\hsize}
        \centering
        \includegraphics[keepaspectratio,scale=.30]{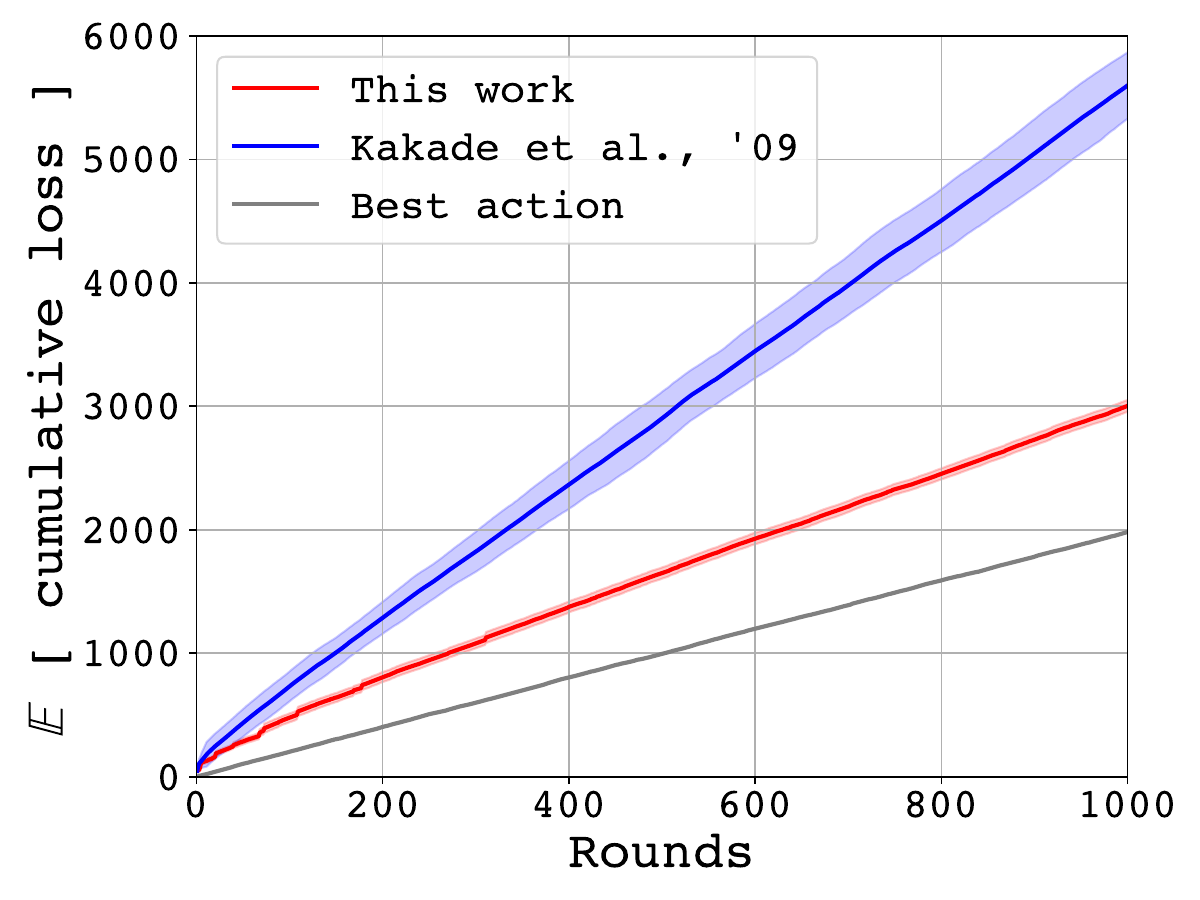}
    \end{minipage}
    \caption{%
        Comparison of previous works and our work. %
        Both figure uses combinatorial vector of dimension $n = 20$.
        Left: comparison of running time. %
        Right: comparison of cumulative losses. %
        KKL '09 is the algorithm %
        proposed by Kakade et al., 2009~\cite{kakade+:sicomp09}. %
        In this experiment, we set $|\combset| = 100$.
    }
    \label{fig:experiments}
\end{figure}
The algorithm proposed by Garber~\cite{garber:mor20}
requires solving large-scale quadratic programs
with $O(n^{2} \sqrt{T} \ln T)$ variables 
for $O(\sqrt{T})$-regret and $O(n^{2} \ln T)$ variables
for $O(T^{2/3})$-regret.
Therefore, one cannot compute a combinatorial vector
even though the number of oracle calls is sub-linear in $T$.
For this reason, we compared our algorithm
against to the algorithm by Kakade et al.~\cite{kakade+:sicomp09}
that achieves $O(\sqrt{T})$-regret\footnote{
    The program is available %
    at \url{https://github.com/rmitsuboshi/metarounding_mitsuboshi}.%
}.
We performed our experiment on Intel Xeon Gold 6124 CPU 2.60GHz processors.

Our experiment was demonstrated on a set cover instance,
generated with the same procedure in~\cite{fujita+:alt13}.
We choose a fixed item size $m = 10$ and
measure the running time by varying the number of sets
$n \in \{10, 50, 100, 200, 500, 1000\}$.
The input for metarounding is a solution of relaxed set cover problem
with a random cost vector $\bm{x} \in [0, 1]^{n}$.
Here, the relaxation means that the variable takes value in $[0, 1]$
instead of $\{0, 1\}$.

\emph{Running time.}
We compared the algorithm proposed by Fujita et al.~\cite{fujita+:alt13} 
and our algorithm.
We use the Gurobi optimizer 9.0.1\footnote{\url{https://www.gurobi.com/}}
to solve the sub-problems.
To solve the entropy minimization, 
we used the sequential quadratic minimization approach;
it first minimizes the quadratic approximation of the objective,
then the objective function is updated to the approximation 
around the optimal solution. This procedure repeats until convergent.
As shown in the left side of Figure~\ref{fig:experiments},
our algorithm is extremely fast compared to the others.

\emph{Cumulative loss.}
We compared the algorithm proposed by Kakade et al.~\cite{kakade+:sicomp09}
and our algorithm.
In each round, the loss vector is generated randomly.
Our algorithm uses 
the classic online gradient descent algorithm~\cite{hazan:mit19}
over $[0, 1]^n$ as the online linear optimization.
The algorithm achieves $R_{T}(1) = O(\sqrt{T})$ over $\relaxset{\combset}$,
so combining it with our metarounding achieves $R_{T}(\alpha) = O(\sqrt{T})$.
As in the right side of Figure~\ref{fig:experiments},
our algorithm is empirically superior to the other.

%% file: conclusion.tex
\section{Conclusion}
In this paper,
we proposed a metarounding algorithm
that finds an $\epsilon$-approximate solution 
in $O( \diamcomb^{2} \ln(n) / \epsilon^{2} )$ rounds.
Even in the numerical experiments,
the proposed algorithm is faster than previous works.

%% file: acknowledgement.tex
\section*{Acknowledgements}
The work was supported by JSPS KAKENHI Grant Numbers JP19H014174 and JP19H04067, respectively.